\documentclass[10pt,a4paper]{amsart}
\usepackage{amsfonts}
\usepackage{amsmath}
\usepackage{amssymb}
\usepackage{amsthm}
\usepackage{braket}
\usepackage{cite}
\usepackage{color}
\usepackage{geometry}
\usepackage{graphicx}
\usepackage{hyperref}
\usepackage{mathrsfs}
\usepackage{mathtools}
\usepackage{setspace}
\usepackage{stmaryrd}
\usepackage{tikz}

\newgeometry{top=3cm,bottom=3cm,outer=2cm,inner=2cm}

\newcommand{\bse}{\begin{subequations}}
\newcommand{\ese}{\end{subequations}}

\newtheorem{theorem}{Theorem}[section]
\newtheorem{definition}[theorem]{Definition}

\newtheorem{remark}[theorem]{Remark}

\newtheorem{corollary}[theorem]{Corollary}
\newtheorem{proposition}[theorem]{Proposition}
\numberwithin{equation}{section}

\DeclareMathOperator{\tr}{tr}
\DeclareMathOperator{\diag}{diag}

\def\undertilde#1{\mathord{\vtop{\ialign{##\crcr
$\hfil\displaystyle{#1}\hfil$\crcr\noalign{\kern1.5pt\nointerlineskip}
$\hfil\tilde{}\hfil$\crcr\noalign{\kern-6.5pt}}}}}
\def\underhat#1{\mathord{\vtop{\ialign{##\crcr
$\hfil\displaystyle{#1}\hfil$\crcr\noalign{\kern1.5pt\nointerlineskip}
$\hfil\hat{}\hfil$\crcr\noalign{\kern-6.5pt}}}}}
\def\underbar#1{\mathord{\vtop{\ialign{##\crcr
$\hfil\displaystyle{#1}\hfil$\crcr\noalign{\kern1.5pt\nointerlineskip}
$\hfil\bar{}\hfil$\crcr\noalign{\kern-6.5pt}}}}}
\def\underdot#1{\mathord{\vtop{\ialign{##\crcr
$\hfil\displaystyle{#1}\hfil$\crcr\noalign{\kern1.5pt\nointerlineskip}
$\hfil\dot{}\hfil$\crcr\noalign{\kern-6.5pt}}}}}
\def\undercheck#1{\mathord{\vtop{\ialign{##\crcr
$\hfil\displaystyle{#1}\hfil$\crcr\noalign{\kern1.5pt\nointerlineskip}
$\hfil\check{}\hfil$\crcr\noalign{\kern-6.5pt}}}}}

\def\wt#1{\tilde{#1}}

\def\ut#1{\undertilde{#1}}

\newcommand{\rT}{\mathrm{T}}
\newcommand{\rd}{\mathrm{d}}

\newcommand{\bLd}{\mathbf{\Lambda}}
\newcommand{\tbLd}{{}^{t\!}\mathbf{\Lambda}}

\newcommand{\bO}{\mathbf{O}}
\newcommand{\tbO}{{}^{t\!}\mathbf{O}}
\newcommand{\bOa}{\mathbf{\Omega}}

\newcommand{\bC}{\mathbf{C}}

\newcommand{\bU}{\mathbf{U}}

\newcommand{\bu}{\mathbf{u}}
\newcommand{\tbu}{{}^{t\!}\mathbf{u}}

\newcommand{\bc}{\mathbf{c}}
\newcommand{\tbc}{{}^{t\!}\mathbf{c}}
\newcommand{\bA}{\mathbf{A}}

\newcommand{\bM}{\mathbf{M}}

\newcommand{\ld}{\lambda}

\newcommand{\oa}{\omega}
\newcommand{\Oa}{\Omega}

\newcommand{\br}{\mathbf{r}}
\newcommand{\bs}{\mathbf{s}}
\newcommand{\bK}{\mathbf{K}}

\title[Direct linearisation of the discrete-time two-dimensional Toda lattices]{Direct linearisation of the discrete-time \\ two-dimensional Toda lattices}
\author[Wei Fu]{Wei Fu}
\address{School of Mathematics, University of Leeds, Leeds LS2 9JT, UK}

\begin{document}

\begin{abstract}
The discrete-time two-dimensional Toda lattice of $A_\infty$-type is studied within the direct linearisation framework, 
which allows us to deal with several nonlinear equations in this class simultaneously and to construct more general solutions of these equations. 
The periodic reductions of this model are also considered, 
giving rise to the discrete-time two-dimensional Toda lattices of $A_{r-1}^{(1)}$-type for $r\geq 2$ 
(which amount to the negative flows of members in the discrete Gel'fand--Dikii hierarchy) and their integrability properties. 

\end{abstract}

\maketitle

\section{Introduction}\label{S:Intro}

The two-dimensional Toda lattice (2DTL) 
\begin{align}\label{2DTL}
 \partial_1\partial_{-1}\varphi_n=-e^{\varphi_{n+1}-\varphi_n}+e^{\varphi_n-\varphi_{n-1}},
\end{align}
where $\varphi_n$ is the potential as a function of two continuous time variables $x_1$ and $x_{-1}$ and one discrete spatial variable $n$, 
and $\partial_1$ and $\partial_{-1}$ denote the partial derivatives $\partial/\partial x_1$ and $\partial/\partial x_{-1}$, respectively, 
was originally proposed by Mikhailov \cite{Mik79} in 1979, as an integrable generalisation of the famous one-dimensional Toda lattice \cite{Tod67}. 
Equation \eqref{2DTL} can alternatively be written in the form of (cf. e.g. Jimbo and Miwa's review paper \cite{JM83}) 
\begin{align}\label{2DTL:Cartan}
 \partial_{1}\partial_{-1}\theta_n=-\sum_{m\in\mathbb{Z}}a_{n,m}e^{-\theta_m},
\end{align}
where $\theta_n=\varphi_{n-1}-\varphi_n$ and $a_{n,m}$ are the entries of the Cartan matrix 
corresponding to the infinite-dimensional algebra $A_\infty$, namely 
\begin{align*}
 a_{n,m}=
 \left\{
 \begin{array}{ll}
  2, & n=m, \\
  -1, & n=m\pm 1, \\
  0, & \hbox{otherwise},
 \end{array}
 \right.
\end{align*}
and thus, it is also reasonable to refer to this model as the 2DTL of $A_\infty$-type. 
The 2DTL is mathematically remarkable because it is still integrable when the algebra $A_\infty$ is replaced by various algebras, 
see Mikhailov, Olshanetsky and Perelomov \cite{MOP83}, Fordy and Gibbons \cite{FG80,FG83}, Wilson \cite{Wil81}, and also the Kyoto School \cite{UT83,DJM83,JM83}. 
Among these the two typical ones are the 2DTLs of $A_1^{(1)}$-type and $A_2^{(2)}$-type which can be written in scalar form, 
namely the well-known sinh--Gordon and Tzitzeica equations 
\begin{align}\label{sinh-Gordon}
 \partial_1\partial_{-1}\varphi_0=e^{2\varphi_0}+e^{-2\varphi_0} \quad \hbox{and} \quad 
 \partial_1\partial_{-1}\varphi_0=e^{2\varphi_0}-e^{-\varphi_0}.
\end{align}
Exact solutions to the 2DTLs of different types are also obtained by using various methods, 
see e.g. \cite{Mik81,HOS88,NW97} for soliton solutions, and also \cite{BMW17} for higher-rank solutions. 
For the algebraic structure of the 2DTL, we refer to Ueno and Takasaki \cite{UT83} (see also \cite{Tak18}). 

Discrete integrable systems are often considered being master models in the theory of integrable systems, 
mainly due to their rich algebraic structure and connection with modern mathematics and physics, see e.g. \cite{HJN16}. 
The discretisation (discrete in both time variables) of the 2DTL was considered from several perspectives in the past decades. 
An effective way to construct the discrete 2DTL is to discretise the bilinear 2DTL, attributed to Hirota \cite{Hir81} and Miwa \cite{Miw82}. 
The mechanism of such a discretisation was explained by Date, Jimbo and Miwa \cite{DJM82} from the viewpoint of the theory of transformation groups, 
motivated by the idea in \cite{Miw82}. 
Moreover, the $q$-discretisation of the 2DTL, which is closely related to quantum groups, was also obtained by the bilinear approach in \cite{KOS94}. 
Another approach to discretise the 2DTL was introduced by Fordy and Gibbons \cite{FG80,FG83}, where the discrete equation arises as 
the superposition formula of two B\"acklund transforms for the continuous-time 2DTL on the nonlinear level. 
A recent significant progress on the discrete-time 2DTLs was made in \cite{GHY12,Smi15}, 
where the integrability is guaranteed in sense that first integrals and symmetries are preserved in the discretisation. 

Our motivation is to consider the discrete-time 2DTL within the so-called direct linearisation framework. 
The direct linearising method was first proposed by Fokas and Ablowitz \cite{FA81} to construct  
several classes of exact solutions of the continuous Korteweg--de Vries (KdV) equation, and was later 
generalised by the Dutch School to study discrete integrable systems, see e.g. \cite{NCW85,NPCQ92,NC95}. 
The benefit of this approach is that it helps us to deal with several discrete nonlinear equations 
in the same class simultaneously, and more importantly, it provides very general solutions to these nonlinear integrable equations. 
Recently, the direct linearising method was developed to study the linear, bilinear and nonlinear structures of a class of integrable equations 
from a unified point of view \cite{FN17a}, illustrated by three-dimensional (3D) models including the discrete Kadomtsev--Petviashvili (KP)-type equations. 

In the present paper, three nonlinear discrete equations in the class of the discrete-time 2DTL of $A_\infty$-type are constructed in the direct linearisation, 
including the bilinear and modified discrete-time 2DTL equations given by Date, Jimbo and Miwa in \cite{DJM82}, and also an unmodified equation 
which is a new parametrisation of an unnamed octahedron-type equation that appeared recently in \cite{FNR15}. 
Their associated Lax pairs are also obtained in the scheme, which completes the incomplete part in \cite{DJM82} 
(though the author firmly believes that those authors were able to construct the Lax pairs in their framework). 
The direct linearising solutions for these equations are natural consequences of the scheme, 
reducing to soliton solutions as a particular degenerate case. 
The periodic reductions are also considered in the framework, resulting in the corresponding two-dimensional (2D) discrete integrable systems, 
namely the discrete-time 2DTLs of $A_{r-1}^{(1)}$-type for $r\geq2$ 
(amounting to the negative flows of members in the discrete Gel'fand--Dikii (GD) hierarchy), 
in which the $A_1^{(1)}$ class provides the well-known discrete sinh--Gordon (or sine--Gordon) equation and its gauge equivalent models. 
When $r\geq3$, we obtain the new discrete integrable systems very recently given by Fordy and Xenitidis \cite{FX17}. 
The nonlinear forms, the bilinear formalism, the Lax matrices and the direct linearising solutions for these models are all discussed, 
from the perspective of the direct linearisation.

The paper is organised as follows. In section \ref{S:DL}, we introduce the general theory of the direct linearisation, 
in the language of infinite matrix. Sections \ref{S:2DTL} and \ref{S:Reduction} are contributed to the direct linearisation schemes 
of the discrete-time 2DTL equations of $A_\infty$-type and $A_{r-1}^{(1)}$-type, respectively, and their integrability properties. 

\section{General theory}\label{S:DL}

\subsection{Infinite-dimensional projection and index-raising matrices}

In the direct linearising method, we need the notion of infinite matrix. 
In this subsection, we only give a brief introduction to such a notion. 
For more detail, we refer the reader to \cite{FN17a}. 
The fundamental objects in infinite matrices include a projection matrix $\bO$ and two index-raising matrices $\bLd$ and $\tbLd$. 
They are all matrices of size $\infty\times\infty$ having their respective entries 
\begin{align}\label{InfMat}
 (\bO)_{i,j}=\delta_{i,0}\delta_{0,j}, \quad (\bLd)_{i,j}=\delta_{i+1,j}, \quad \hbox{and} \quad (\tbLd)_{i,j}=\delta_{i,j+1},
\end{align}
respectively, where $\delta_{\cdot,\cdot}$ is the standard Kronecker delta function defined as 
\begin{align*}
 \delta_{i,j}=
 \left\{
 \begin{array}{ll}
 1, & i=j, \\
 0, & i\neq j,
 \end{array}
 \right. \quad
 \forall i,j\in\mathbb{Z}.
\end{align*}
From the definitions, one can observe that 
the transpose of $\bO$ obeys $\tbO=\bO$, and $\bLd$ and $\tbLd$ are the transposes for each other. 
The projection operator $\bO$ is a ``centred'' infinite matrix and satisfies the property $\bO^2=\bO$. 
Given a general infinite matrix $\bU$ with the $(i,j)$-entries $U_{i,j}$, one can verify that 
\begin{align*}
 (\bO\cdot\bU)_{i,j}=\delta_{i,0}U_{0,j} \quad \hbox{and} \quad (\bU\cdot\bO)_{i,j}=U_{i,0}\delta_{0,j},
\end{align*}
according to the definition of $\bO$. In terms of $\bLd$ and $\tbLd$, we have operations 
\begin{align*}
 (\bLd^{i'}\cdot\bU)_{i,j}=U_{i+i',j} \quad \hbox{and} \quad (\bU\cdot\tbLd^{j'})_{i,j}=U_{i,j+j'},
\end{align*}
which follow from the definitions of $\bLd$ and $\tbLd$, namely they raise the row and column indices, respectively. 
Similarly, for a given infinite-dimensional column vector $\bu$ having its $i$th-component $u^{(i)}$, and its transpose, 
i.e. an infinite-dimensional row vector $\tbu$ still having its $i$th-component $u^{(i)}$, we have the following operations: 
\begin{align*}
 (\bO\cdot\bu)_{i}=\delta_{i,0}u^{(0)}, \quad (\tbu\cdot\bO)_{j}=u^{(0)}\delta_{0,j}, \quad 
 (\bLd^{i'}\cdot\bu)_{i}=u^{(i+i')}, \quad (\bu\cdot\tbLd^{j'})_{j}=u^{(j+j')},
\end{align*}
where $(\cdot)_i$ and $(\cdot)_j$ denote taking the $i$th- and $j$th-components of an infinite-dimensional vector, respectively. 
In the direct linearisation, we need to consider two particular infinite-dimensional vectors 
\begin{align*}
 \bc_k=(\cdots,k^{-1},1,k,\cdots)^\rT \quad \hbox{and} \quad
 \tbc_{k'}=(\cdots,k'^{-1},1,k',\cdots).
\end{align*}
The operations of $\bO$, $\bLd$ and $\tbLd$ on these two vectors satisfy the following relations: 
\begin{align}\label{P:ck}
 \bLd^i\cdot\bc_k=k^i\bc_k, \quad \tbc_{k'}\cdot\tbLd^j=k'^j\tbc_{k'}, \quad \tbc_{k'}\cdot\bO\cdot\bc_k=1.
\end{align}
The above are the fundamental operations of infinite-dimensional matrices and vectors. 
More complicated operations involving these objects can be derived from these fundamental ones or from the definitions. 

\subsection{Linear integral equation and its infinite matrix representation}

The idea of the direct linearising method is to solve a nonlinear integrable equation by considering its associated linear integral equation. 
Reversely, a linear integral equation having a certain structure also brings a class of related nonlinear equations. 

We consider a general linear integral equation taking the form of 
\begin{align}\label{Integral}
 \bu_k+\iint_D\rd\zeta(l,l')\rho_k\Oa_{k,l'}\sigma_{l'}\bu_l=\rho_k\bc_k, 
\end{align}
where the wave function $\bu_k$ is an infinite-dimensional column vector having its $i$th-component $u_k^{(i)}$ as a smooth function of 
the dynamical variables and their associated parameters, also depending on the spectral parameter $k$; 
the Cauchy kernel $\Oa_{k,l'}$ is an algebraic expression of the spectral parameters $k$ and $l'$, which is independent of the dynamical variables; 
the plane wave functions $\rho_k$ and $\sigma_{l'}$ are expressions of dynamical variables (i.e. flow variables), 
depending on the spectral parameters $k$ and $l'$, respectively; 
the measure $\rd\zeta(l,l')$ depending on the spectral variables $l$ and $l'$, and the integration domain $D$ can be determined later for particular 
classes of solutions. 

The key point is that the plane wave factors $\rho_k$ and $\sigma_{k'}$ containing the information of the dynamics, 
determine the linear structure in the integral equation \eqref{Integral}; 
the Cauchy kernel $\Oa_{k,k'}$ and the measure $\rd\zeta(k,k')$, instead, describe the dependence of the spectral variables $k$ and $k'$, 
governing the nonlinear structure of the resulting nonlinear equations.  

In order to relate a class of nonlinear equations to a given linear integral equation, 
we interpret the structure of the linear integral equation \eqref{Integral} in infinite matrix language.
First of all, we introduce an infinite matrix $\bOa$ defined by the following relation: 
\begin{align}\label{Omega}
 \Oa_{k,k'}=\tbc_{k'}\cdot\bOa\cdot\bc_k.
\end{align}
This relation implies that $\bOa$ is the infinite matrix representation of the Cauchy kernel $\Oa_{k,k'}$. 
The second object we need is an infinite matrix $\bC$ defined as 
\begin{align}\label{C}
 \bC=\iint_D\rd\zeta(k,k')\rho_k\bc_k\tbc_{k'}\sigma_{k'}.
\end{align}
The core part of this infinite matrix is $\rho_k\sigma_{k'}$, namely the effective dispersion in the linear integral equation, 
which implies that $\bC$ is the infinite matrix representation of the effective plane wave factor. 
We also define an infinite matrix $\bU$ by 
\begin{align}\label{Potential}
 \bU=\iint_D\rd\zeta(k,k')\bu_k\tbc_{k'}\sigma_{k'},
\end{align}
which is a nonlinearisation of the wave function $\bu_k$ in the linear integral equation \eqref{Integral}. 
Using these quantities, we can represent the linear integral equation by the above infinite matrices. 
\begin{proposition}
The linear integral \eqref{Integral} has an infinite matrix representation 
\begin{align}\label{uk}
 \bu_k=(1-\bU\cdot\bOa)\cdot\rho_k\bc_k,
\end{align}
where $\bOa$, $\bU$ are the infinite matrices defined in \eqref{Omega} and \eqref{Potential}. 
\end{proposition}
\begin{corollary}
The infinite matrix $\bU$ satisfies the following relation: 
\begin{align}\label{U}
\bU=(1-\bU\cdot\bOa)\cdot\bC, \quad \hbox{or alternatively} \quad \bU=\bC\cdot(1+\bOa\cdot\bC)^{-1};
\end{align}
in other words, the dynamics of $\bU$ is governed by $\bC$ (the linear structure) and $\bOa$ (the nonlinear structure). 
\end{corollary}
In addition, we also need the notion of the tau function in the framework. 
\begin{definition}
The tau function in the direct linearisation framework is defined as 
\begin{align}\label{tau}
 \tau=\det(1+\bOa\cdot\bC), 
\end{align}
in which the determinant should be understood as the expansion of $\exp\{\tr\left[\ln(1+\bOa\cdot\bC)\right]\}$. 
\end{definition}
These quantities constitute the key ingredients in the direct linearising approach. 
For a given linear integral equation, namely for fixed plane wave factors $\rho_k$ and $\sigma_{k'}$ (i.e. the linear structure), 
and fixed Cauchy kernel $\Oa_{k,k'}$ and measure $\rd\zeta(k,k')$ (i.e. the nonlinear structure), 
the direct linearisation helps to algebraically construct a class of closed-form equations from the following perspectives: 
i) The components of $\bu_k$ solve linear equations (i.e. Lax pairs); 
ii) The tau function $\tau$ solves the bilinear equations in Hirota's form; 
iii) The entries of $\bU$ solve the nonlinear equations in the class; 
iv) The quantities given in \eqref{Integral}, \eqref{tau} and \eqref{Potential} provide the direct linearising solutions 
to the corresponding closed-form equations.

\section{Discrete-time two-dimensional Toda lattice of $A_\infty$-type}\label{S:2DTL}

\subsection{Notations}\label{S:Notation}
We introduce some essential notations in the discrete theory. For an arbitrary function $f=f_{n_1,n_{-1},n}\doteq f(n_1,n_{-1},n)$ of 
discrete independent variables $n_1$, $n_{-1}$ and $n$ associated with lattice parameters $p_1$, $p_{-1}$ and zero, respectively. 
The following accent marks denote the discrete forward shift operators with respect to their corresponding lattice directions: 
\begin{align*}
 &\wt f=f_{n_1+1,n_{-1},n}, \quad \check f=f_{n_1,n_{-1}+1,n}, \quad \dot f=f_{n_1,n_{-1},n+1}. 
\end{align*}
Similarly, for backward shift operators we have 
\begin{align*}
 &\ut f=f_{n_1-1,n_{-1},n}, \quad \undercheck f=f_{n_1,n_{-1}-1,n}, \quad \underdot f=f_{n_1,n_{-1},n-1}. 
\end{align*}
Combinations of the above accent marks then naturally denote the compositions of these discrete shift operations.

\subsection{Nonlinear and linear structures}\label{S:Algebra}

To algebraically construct equations in the class of the discrete-time 2DTL of $A_\infty$-type, we consider the following Cauchy kernel and measure: 
\begin{align}\label{dt2DTL:Kernel}
 \Oa_{k,k'}=\frac{1}{k+k'}, \quad \rd\zeta(k,k') \quad \hbox{being arbitrary},
\end{align}
as well as the plane wave factors 
\begin{align}\label{dt2DTL:Linear}
 \rho_k=(p_1+k)^{n_1}(p_{-1}+k)^{n_{-1}}k^n \quad \hbox{and} \quad \sigma_{k'}=(p_1-k')^{-n_1}(p_{-1}-k'^{-1})^{-n_{-1}}(-k')^{-n}.
\end{align}
In concrete calculation, the two plane wave factors normally combine together (cf. the infinite matrix $\bC$), and form the effective plane wave factor 
\begin{align}\label{dt2DTL:PWF}
 \rho_k\sigma_{k'}=\left(\frac{p_1+k}{p_1-k'}\right)^{n_1}\left(\frac{p_{-1}+k^{-1}}{p_{-1}-k'^{-1}}\right)^{n_{-1}}\left(-\frac{k}{k'}\right)^n,
\end{align}
which completely governs the dynamics in the resulting nonlinear models. 

\subsection{Soliton solution}\label{S:Soliton}
According to the general statement, equations \eqref{Integral}, \eqref{tau} and \eqref{Potential} 
provide us with the direct linearising solutions to the discrete-time 2DTL of $A_\infty$-type. 
As an example, below we give the explicit formulae for soliton solutions. 
We take a particular measure involving a finite number of singularities, namely 
\begin{align}\label{dt2DTL:Sigularity}
 \rd\zeta(k,k')=\sum_{i=1}^{N}\sum_{j=1}^{N'}A_{i,j}\delta(k-k_i)\delta(k'-k'_j)\rd k\rd k',
\end{align}
where $A_{i,j}$ is the $(i,j)$-entry of a constant full-rank $N\times N'$ matrix $\bA$, $k_i$ and $k'_j$ are the singularities, 
and the delta function here should be understood as $\delta(k-k_j)=\frac{1}{2\pi\mathfrak{i}}\frac{1}{k-k_j}$ with $\mathfrak{i}$ being the imaginary unit. 
Let the domain $D$ contain all these singularities. This reduces the linear integral equation \eqref{Integral} to 
\begin{align}\label{Integral:Reduc}
 \bu_k+\sum_{i=1}^{N}\sum_{j=1}^{N'}A_{i,j}\frac{\rho_k\sigma_{k'_j}}{k+k'_j}\bu_{k_i}=\rho_k\bc_k,
\end{align}
according to the residue theorem. Now we introduce an $N'\times N$ generalised Cauchy matrix $\bM$ with its $(j,i)$-entry 
\begin{align*}
 M_{j,i}=\frac{\rho_{k_i}\sigma_{k'_j}}{k_i+k'_j}, \quad j=1,2,\cdots,N', \quad i=1,2,\cdots,N,
\end{align*}
where $\rho_k\sigma_{k'}$ is given by \eqref{dt2DTL:PWF}. 
Taking $k=k_i$ for $i=1,2,\cdots,N$ in the reduced equation \eqref{Integral:Reduc}, we obtain 
\begin{align*}
 (\bu_{k_1},\bu_{k_2},\cdots,\bu_{k_N})+(\bu_{k_1},\bu_{k_2},\cdots,\bu_{k_N})\bA\bM
 =(\rho_{k_1}\bc_{k_1},\rho_{k_2}\bc_{k_2},\cdots,\rho_{k_N}\bc_{k_N})
\end{align*}
and thus, $\bu_{k_i}$ can be expressed by 
\begin{align*}
 (\bu_{k_1},\bu_{k_2},\cdots,\bu_{k_N})=\br^\rT\diag(\bc_{k_1},\bc_{k_2},\cdots,\bc_{k_N})(1+\bA\bM)^{-1},
\end{align*}
where $\br=(\rho_{k_1},\rho_{k_2},\cdots,\rho_{k_N})^\rT$. 
Meanwhile, considering \eqref{Potential} together with \eqref{dt2DTL:Sigularity}, we also have 
\begin{align}\label{Potential:Reduc}
 \bU=\sum_{i=1}^{N}\sum_{j=1}^{N'}A_{i,j}\bu_{k_i}\tbc_{k'_j}\sigma_{k'_j}=(\bu_{k_1},\bu_{k_2},\cdots,\bu_{k_N})\bA
 \diag(\bc_{k'_1},\bc_{k'_2},\cdots,\bc_{k'_{N'}})\bs.
\end{align}
where $\bs=(\sigma_{k'_1},\sigma_{k'_2},\cdots,\sigma_{k'_{N'}})^\rT$. 
Therefore, substituting $\bu_{k_i}$ in \eqref{Potential:Reduc} by using the previous formula, we end up with the explicit expression of $\bU$ as follows: 
\begin{align*}
 \bU=\br^\rT\diag(\bc_{k_1},\bc_{k_2},\cdots,\bc_{k_N})(1+\bA\bM)^{-1}\bA\diag(\bc_{k'_1},\bc_{k'_2},\cdots,\bc_{k'_{N'}})\bs.
\end{align*}
Furthermore, we can also consider the soliton expression for the tau function. According to the definition of the determinant of an infinite matrix, 
we have the following relation for the tau function: 
\begin{align*}
 \ln\tau=\ln\left(\det(1+\bOa\cdot\bC)\right)=\tr\left(\ln(1+\bOa\cdot\bC)\right).
\end{align*}
Thus, by expansion the right hand side can be written as 
\begin{align*}
 \tr\left(\sum_{\gamma=1}^\infty(-1)^{\gamma+1}\frac{1}{\gamma}(\bOa\cdot\bC)^\gamma\right)
 =\sum_{\gamma=1}^\infty(-1)^{\gamma+1}\frac{1}{\gamma}\tr(\bOa\cdot\bC)^\gamma.
\end{align*}
Considering the degeneration \eqref{dt2DTL:Sigularity} and the definition of $\bC$, i.e. \eqref{C}, we can derive
\begin{align*}
 \tr(\bOa\cdot\bC)^\gamma
 =\tr\left(\bOa\cdot\left(\sum_{i=1}^{N}\sum_{j=1}^{N'}A_{i,j}\rho_{k_i}\bc_{k_i}\tbc_{k'_j}\sigma_{k'_j}\right)\right)^\gamma=\tr(\bA\bM)^\gamma,
\end{align*}
where the cyclic permutation of the trace is used in the last step. 
Thus, we end up with the relation 
\begin{align*}
 \ln\tau&=\sum_{\gamma=1}^\infty(-1)^{\gamma+1}\frac{1}{\gamma}\tr(\bA\bM)^\gamma
 =\tr\left(\sum_{\gamma=1}^\infty(-1)^{\gamma+1}\frac{1}{\gamma}(\bA\bM)^\gamma\right) \\
 &=\tr\left(\ln(1+\bA\bM)\right)=\ln\left(\det(1+\bA\bM)\right). 
\end{align*}
We have derived the soliton formulae for the linear, nonlinear and bilinear variables. 
For convenience in the future, we give the formulae for the components $u_k^{(i)}$ in the wave function, 
the entries $U_{i,j}$ in the infinite matrix $\bU$ as well as the tau function $\tau$, 
and conclude all the results in the following theorem: 
\begin{theorem}\label{T:dt2DTLSol}
For the Cauchy kernel, the measure and the effective plane wave factor given in \eqref{dt2DTL:Kernel}, \eqref{dt2DTL:Sigularity} and \eqref{dt2DTL:PWF}, 
the wave function, the tau function and the potential variable in the $(N,N')$-soliton form are determined by the following: 
\begin{align*}
 (u_{k_1}^{(i)},u_{k_2}^{(i)},\cdots,u_{k_N}^{(i)})=\br^\rT\bK^i(1+\bA\bM)^{-1}, \quad
 \tau=\det(1+\bA\bM), \quad
 U_{i,j}=\br^\rT\bK^i(1+\bA\bM)^{-1}\bA\bK'^j\bs,
\end{align*}
in which $\bK=\diag(k_1,k_2,\cdots,k_N)$ and $\bK'=\diag(k'_1,k'_2,\cdots,k'_{N'})$. 
\end{theorem}

\subsection{Discrete dynamics}\label{S:Dynamics}

In this subsection, we discuss the infinite matrix formalism in the direct linearisation 
based on the nonlinear and linear structures of the discrete-time 2DTL, i.e. \eqref{dt2DTL:Kernel} and \eqref{dt2DTL:PWF}, 
which will be used to construct closed-form equations in the next subsection. 

Equation \eqref{dt2DTL:Kernel} implies that the Cauchy kernel obeys the relation $\Oa_{k,k'}k+k'\Oa_{k,k'}=1$. 
Recalling the definition of $\bOa$, i.e. \eqref{Omega}, we can deduce that 
\bse\label{dt2DTL:OaDyn}
\begin{align}
 \bOa\cdot\bLd+\tbLd\cdot\bOa=\bO,
\end{align}
and consequently it can be generalised to 
\begin{align}
 \bOa\cdot(p_1+\bLd)-(p_1-\tbLd)\cdot\bOa=\bO \quad \hbox{and} \quad \bOa\cdot(p_{-1}+\bLd^{-1})-(p_{-1}-\tbLd^{-1})\cdot\bOa=\tbLd^{-1}\cdot\bO\cdot\bLd^{-1}.
\end{align}
\ese
The above relations form the nonlinear structure of the infinite matrix formalism. Below we give the linear structure. 
Observing that the effective plane wave factor \eqref{dt2DTL:PWF} satisfies 
\begin{align*}
 (\rho_k\sigma_{k'})^{\tilde{}}=\left(\frac{p_1+k}{p_1-k'}\right)\rho_k\sigma_{k'}, \quad 
 (\rho_k\sigma_{k'})^{\check{}}=\left(\frac{p_{-1}+k^{-1}}{p_{-1}-k'^{-1}}\right)\rho_k\sigma_{k'}, \quad 
 (\rho_k\sigma_{k'})^{\dot{}}=\left(-\frac{k}{k'}\right)\rho_k\sigma_{k'},
\end{align*}
and the property of $\bc_k$ given in \eqref{P:ck}, we can derive from \eqref{C} the dynamical evolutions of $\bC$ as follows: 
\begin{align}\label{dt2DTL:CDyn}
 \wt\bC\cdot(p_1-\tbLd)=(p_1+\bLd)\cdot\bC, \quad \check\bC\cdot(p_{-1}-\tbLd^{-1})=(p_{-1}+\bLd^{-1})\cdot\bC, \quad \dot\bC\cdot(-\tbLd)=\bLd\cdot\bC.
\end{align}
Equations \eqref{dt2DTL:OaDyn} and \eqref{dt2DTL:CDyn} are the fundamental relations 
which can help to build up the dynamics for key ingredients, namely $\bU$, $\bu_k$ and $\tau$, in the scheme.
\begin{proposition}
The infinite matrix $\bU$ obeys the following dynamical evolutions with respect to the lattice variables $n_1$, $n_{-1}$ and $n$: 
\bse\label{dt2DTL:UDyn}
\begin{align}
 &\wt\bU\cdot(p_1-\tbLd)=(p_1+\bLd)\cdot\bU-\wt\bU\cdot\bO\cdot\bU, \label{dt2DTL:UDyna} \\
 &\check\bU\cdot(p_{-1}-\tbLd^{-1})=(p_{-1}+\bLd^{-1})\cdot\bU-\check\bU\cdot\tbLd^{-1}\cdot\bO\cdot\bLd^{-1}\cdot\bU, \label{dt2DTL:UDynb} \\
 &\dot\bU\cdot(-\tbLd)=\bLd\cdot\bU-\dot\bU\cdot\bO\cdot\bU, \label{dt2DTL:UDync}
\end{align}
\ese
which are the fundamental relations for constructing nonlinear equations in the class of the 2DTL of $A_\infty$-type. 
\end{proposition}
\begin{proof}
We only prove the first equation. According to the infinite matrix representation \eqref{U}, we have 
\begin{align*}
 \wt\bU\cdot(p_1-\tbLd)=(1-\wt\bU\cdot\bOa)\cdot\wt\bC\cdot(p_1-\tbLd).
\end{align*}
In light of the first dynamical relation in \eqref{dt2DTL:CDyn}, this relation can be written as 
\begin{align*}
 \wt\bU\cdot(p_1-\tbLd)=(1-\wt\bU\cdot\bOa)\cdot(p_1+\bLd)\cdot\bC=(p_1+\bLd)\cdot\bC-\wt\bU\cdot\bOa\cdot(p_1+\bLd)\cdot\bC.
\end{align*}
Notice the second relation in \eqref{dt2DTL:OaDyn}. The above equation can further be reformulated as 
\begin{align*}
 \wt\bU\cdot(p_1-\tbLd)=(p_1+\bLd)\cdot\bC-\wt\bU\cdot\left[\bO+(p_1-\tbLd)\cdot\bOa\right]\cdot\bC,
\end{align*}
which amounts to 
\begin{align*}
 \wt\bU\cdot(p_1-\tbLd)\cdot(1+\bOa\cdot\bC)=(p_1+\bLd)\cdot\bC-\wt\bU\cdot\bO\cdot\bC.
\end{align*}
Multiplying this equation by $(1+\bOa\cdot\bC)^{-1}$ from the right and recalling $\bU=\bC\cdot(1+\bOa\cdot\bC)^{-1}$, 
we obtain the first equation in \eqref{dt2DTL:UDyn}. The second and third equations are proven similarly. 
\end{proof}
Equations in \eqref{dt2DTL:UDyn} can further help to derive the dynamics of the wave function $\bu_k$. 
\begin{proposition}
The wave function $\bu_k$ satisfies the following dynamical evolutions with respect to the discrete flow variables $n_1$, $n_{-1}$ and $n$: 
\bse\label{dt2DTL:ukDyn}
\begin{align}
 &\wt\bu_k=(p_1+\bLd)\cdot\bu_k-\wt\bU\cdot\bO\cdot\bu_k, \label{dt2DTL:ukDyna} \\
 &\check\bu_k=(p_{-1}+\bLd^{-1})\cdot\bu_k-\check\bU\cdot\tbLd^{-1}\cdot\bO\cdot\bLd^{-1}\cdot\bu_k, \label{dt2DTL:ukDynb} \\
 &\dot\bu_k=\bLd\cdot\bu_k-\dot\bU\cdot\bO\cdot\bu_k, \label{dt2DTL:ukDync}
\end{align}
\ese
which will be used to construct the closed-form linear equations in the class of the 2DTL of $A_\infty$-type. 
\end{proposition}
\begin{proof}
Once again, we only prove the first equation. Forward-shifting equation \eqref{uk} with respect to $n_1$ by one unit gives us 
\begin{align*}
 \wt\bu_k=(1-\wt\bU\cdot\bOa)\cdot\wt\rho_k\bc_k=(1-\wt\bU\cdot\bOa)\cdot(p_1+k)\rho_k\bc_k=(1-\wt\bU\cdot\bOa)\cdot(p_1+\bLd)\cdot\rho_k\bc_k,
\end{align*}
where the last step holds due to \eqref{P:ck}. In virtue of the second relation in \eqref{dt2DTL:OaDyn}, we further obtain 
\begin{align*}
 \wt\bu_k=(p_1+\bLd)\cdot\rho_k\bc_k-\wt\bU\cdot\bOa\cdot(p_1+\bLd)\cdot\rho_k\bc_k
 =(p_1+\bLd)\cdot\rho_k\bc_k-\wt\bU\cdot\left[\bO+(p_1-\tbLd)\cdot\bOa\right]\cdot\rho_k\bc_k.
\end{align*}
With the help of \eqref{dt2DTL:UDyna}, this can be written as 
\begin{align*}
 \wt\bu_k=(p_1+\bLd)\cdot\rho_k\bc_k-\wt\bU\cdot\bO\cdot\rho_k\bc_k-\left[(p_1+\bLd)\cdot\bU-\wt\bU\cdot\bO\cdot\bU\right]\cdot\bOa\cdot\rho_k\bc_k,
\end{align*}
which is nothing but \eqref{dt2DTL:ukDyna} since $\bu_k=(1-\bU\cdot\bOa)\cdot\rho_k\bc_k$ (see \eqref{uk}). 
\end{proof}
In the framework, we also need to consider dynamical evolutions of the tau function with respect to these lattice variables. 
Below we only list two relations which will be used later. 
\begin{proposition}
The tau function obeys evolutions 
\begin{align}\label{dt2DTL:tauDyn}
 \frac{\dot\tau}{\tau}=1-U_{0,-1} \quad \hbox{and} \quad \frac{\underdot\tau}{\tau}=1-U_{-1,0}
\end{align}
with respect to the lattice direction $n$.
\end{proposition}
\begin{proof}
Notice that the tau function is defined by \eqref{tau}. We can calculate that 
\begin{align*}
 \dot\tau=\det(1+\bOa\cdot\dot\bC)=\det\left(1+\bOa\cdot\bLd\cdot\bC\cdot(-\tbLd)^{-1}\right)
 =\det\left(1+(\bO-\tbLd\cdot\bOa)\cdot\bC\cdot(-\tbLd)^{-1}\right),
\end{align*}
where the last equation holds because of \eqref{dt2DTL:OaDyn}. We can then further reformulate this equation as 
\begin{align*}
 \dot\tau&=\det\left(1+\bOa\cdot\bC-\tbLd^{-1}\cdot\bO\cdot\bC\right)
 =\det\left(1+\bOa\cdot\bC\right)\det\left(1-(1+\bOa\cdot\bC)^{-1}\cdot\tbLd^{-1}\cdot\bO\cdot\bC\right) \\
 &=\tau\left[1-(\bC\cdot(1+\bOa\cdot\bC)^{-1}\cdot\tbLd^{-1})_{0,0}\right]=\tau\left[1-(\bU\cdot\tbLd^{-1})_{0,0}\right]=\tau(1-U_{0,-1}),
\end{align*}
where the rank 1 Weinstein--Aronszajn formula is used to evaluate the determinant in terms of a scalar quantity, see \cite{FN17a}. 
This is exactly the first identity for the tau function. The second one can be derived similarly and we skip the proof. 
\end{proof}

\subsection{Closed-form nonlinear equations and associated linear systems}\label{S:Closed}

In order to construct closed-form equations, we introduce the unmodified variable $u$, and the modified variables $v$ and $w$ as follows: 
\begin{align*}
 u\doteq U_{0,0}=(\bU)_{0,0}, \quad v\doteq 1-U_{0,-1}=1-(\bU\cdot\tbLd^{-1})_{0,0}, \quad w\doteq 1-U_{-1,0}=1-(\bLd^{-1}\cdot\bU)_{0,0}.
\end{align*}
These variables and the tau function are connected with each other via certain difference transforms. We list all these relations in the proposition below. 
\begin{proposition}
The unmodified variable $u$ and the modified variable $v$ are related via the following discrete Miura transform: 
\bse
\begin{align}\label{dt2DTL:MT}
 p_1+\dot u-\wt u=p_1\frac{\wt v}{v}, \quad 1+p_{-1}(u-\check u)=\frac{\check v}{\underdot v}.
\end{align}
The bilinear transforms between the nonlinear variables $u$, $v$ and the bilinear variable $\tau$ are given by 
\begin{align}\label{dt2DTL:BLT}
 p_1+\dot u-\wt u=p_1\frac{\tau\dot{\wt\tau}}{\dot\tau\wt\tau}, \quad 1+p_{-1}(u-\check u)=\frac{\underdot\tau\dot{\check\tau}}{\tau\check\tau} \quad 
 \hbox{and} \quad v=\frac{\dot\tau}{\tau},
\end{align}
\ese
respectively. The two modified variables $v$ and $w$ satisfy a simple relation $w=1/\underdot v$. 
\end{proposition}
\begin{proof}
Recalling \eqref{dt2DTL:tauDyn} and the definitions of $v$ and $w$, we can easily observe that $v=\dot\tau/\tau$ and $w=\underdot\tau/\tau$, 
which are the bilinear transforms for the two modified variables, leading to the relation $v\dot w=1$ immediately,\footnote{
Such an identity can alternatively be derived from the infinite matrix relation \eqref{dt2DTL:UDync} by taking the $(-1,-1)$-entry, cf. \cite{NCW85}. 
Here we can observe that this relation is a consequence of a trivial identity based on the tau function. 
} 
i.e. $w=1/\underdot v$. 
Taking the $(0,-1)$-entry of \eqref{dt2DTL:UDyna} and \eqref{dt2DTL:UDync}, respectively, we obtain 
\begin{align*}
 p_1(v-\wt v)=U_{1,-1}+\wt u v \quad \hbox{and} \quad U_{1,-1}+\dot u v=0.
\end{align*}
Subtracting the two relations and eliminating $U_{1,-1}$ gives rise to the first half in the Miura transform \eqref{dt2DTL:MT}. 
Next, we take the $(0,0)$-entry of \eqref{dt2DTL:UDynb}, which results in 
\begin{align}\label{dt2DTL:vw-u}
 1+p_{-1}(u-\check u)=\check v w.
\end{align}
This is exactly the second half of the Miura transform once the variable $w$ is replaced by $1/\underdot v$. 
The bilinear transform between $u$ and $\tau$ can then be derived 
by substituting $v$ in the Miura transform \eqref{dt2DTL:MT} with the tau function via $v=\dot\tau/\tau$. 
\end{proof}

With the help of these transforms, we are able to construct closed-form linear systems based on the eigenfunction $\phi\doteq(\bu_k)_0=u_k^{(0)}$. 
\begin{theorem}\label{T:dt2DTLLax}
In the class of the discrete-time 2DTL of $A_\infty$-type, we have the following linear systems: 
\bse\label{dt2DTL:Lax}
\begin{align}
 &\wt\phi=(p_1+\dot u-\wt u)\phi+\dot\phi=p_1\frac{\wt v}{v}\phi+\dot\phi=p_1\frac{\tau\dot{\wt\tau}}{\dot\tau\wt\tau}\phi+\dot\phi, \label{dt2DTL:Laxa} \\
 &\check\phi=p_{-1}\phi+\left[1+p_{-1}(u-\check u)\right]\underdot\phi=p_{-1}\phi+\frac{\check v}{\underdot v}\underdot\phi
 =p_{-1}\phi+\frac{\underdot\tau\dot{\check\tau}}{\tau\check\tau}\underdot\phi. \label{dt2DTL:Laxb}
\end{align}
\ese
which are the Lax pairs for their corresponding nonlinear integrable lattice equations of $u$, $v$ and $\tau$. 
\end{theorem}
\begin{proof}
We only prove the linear equations involving the unmodified variable $u$. 
The other equalities are natural consequences of this under the difference transforms given in \eqref{dt2DTL:MT} and \eqref{dt2DTL:BLT}. 
The relation \eqref{dt2DTL:ukDync} can be expressed by
\begin{align*}
 \bLd\cdot\bu_k=\dot\bu_k+\dot\bU\cdot\bO\cdot\bu_k, \quad \hbox{or alternatively} \quad 
 \bLd^{-1}\cdot\bu_k=\underdot\bu_k-\bLd^{-1}\cdot\bU\cdot\bO\cdot\underdot\bu_k,
\end{align*}
which helps us to eliminate $\bLd\cdot\bu_k$ and $\bLd^{-1}\cdot\bu_k$ in \eqref{dt2DTL:ukDyna} and \eqref{dt2DTL:ukDynb}, respectively, and obtain 
the following relations involving only $\bu_k$: 
\begin{align*}
 &\wt\bu_k=(p_1+\dot\bU\cdot\bO-\wt\bU\cdot\bO)\cdot\bu_k+\dot\bu_k, \\
 &\check\bu_k=p_{-1}\bu_k+(1-\check\bU\cdot\tbLd^{-1}\cdot\bO)\cdot(1-\bLd^{-1}\cdot\bU\cdot\bO)\cdot\underdot\bu_k
\end{align*}
The $0$th-component of the first equation immediately gives rise to the ``tilde'' equation, i.e. \eqref{dt2DTL:Laxa}. 
For the ``check'' equation in \eqref{dt2DTL:Laxb}, we consider the $0$th-component of the second equation and obtain 
\begin{align*}
 \check\phi=p_{-1}\phi+\check v w\underdot\phi,
\end{align*}
which then turns out to be the second equation \eqref{dt2DTL:Laxb} once $\check v w$ is replaced by $1+p_{-1}(u-\check u)$,
with the help of the identity \eqref{dt2DTL:vw-u}. 
\end{proof}

The compatibility condition of the equations in each Lax pair listed in \eqref{dt2DTL:Lax} gives us a 3D nonlinear integrable difference equation. 
\begin{theorem}\label{T:dt2DTLNL}
The unmodified variable $u$, the unmodified variable $v$ and the tau function $\tau$ solve the following 3D integrable discrete equations:  
\bse\label{dt2DTL:NL}
\begin{align}
 &\frac{p_1+\dot{\check u}-\check{\wt u}}{p_1+u-\underdot{\wt u}}=\frac{1+p_{-1}(\wt u-\check{\wt u})}{1+p_{-1}(u-\check u)}, \label{dt2DTL:Un} \\
 &p_1p_{-1}\left(\frac{\check{\wt v}}{\check v}-\frac{\wt v}{v}\right)=-\frac{\dot{\check v}}{v}+\frac{\check{\wt v}}{\underdot{\wt v}}, \label{dt2DTL:Mod} \\
 &p_1p_{-1}(\tau\check{\wt\tau}-\wt\tau\check\tau)=\tau\check{\wt\tau}-\dot{\check\tau}\underdot{\wt\tau}, \label{dt2DTL:BL}
\end{align}
\ese
which we refer to as the unmodified, modified and bilinear discrete-time 2DTL equations of $A_\infty$-type, respectively. 
The direct linearising solutions for these equations are governed by \eqref{Potential} and \eqref{tau}, respectively. 
The $(N,N')$-soliton solutions are given in theorem \ref{T:dt2DTLSol}. 
\end{theorem}
\begin{proof}
Equations \eqref{dt2DTL:Un} and \eqref{dt2DTL:Mod} follow from the $u$ and $v$ parts of the linear equations \eqref{dt2DTL:Laxa} and \eqref{dt2DTL:Laxb}. 
The compatibility condition $\check{\wt\phi}=\wt{\check\phi}$ in terms of $\tau$ gives a quartic equation 
\begin{align*}
 p_1p_{-1}\left(\tau\check{\wt\tau}\dot{\check\tau}\wt{\dot\tau}-\wt\tau\check\tau\dot\tau\dot{\check{\wt\tau}}\right)
 =\tau\wt\tau\check{\wt\tau}\dot{\dot{\check\tau}}-\underdot{\wt\tau}\dot\tau\dot{\check\tau}\dot{\check{\wt\tau}}.
\end{align*}
This is a weak equation defined on 10 points and by discrete integration we obtain the bilinear equation \eqref{dt2DTL:BL}. 
\end{proof}

\begin{remark}
Making use of the transform $v=1/\underdot w$, we can derive another modified equation expressed by $w$ which is dual to \eqref{dt2DTL:Mod}, 
cf. \eqref{dt2DTL:tauDyn}, and its Lax pair can be obtained from \eqref{dt2DTL:Lax} under the same transform. 
These is also a closed-form equation based on the Schwarzian variable $z\doteq U_{-1,-1}$, which takes a dual form of the unmodified equation \eqref{dt2DTL:Un}. 
This is because in the discrete-time 2DTL the effective plane wave factor \eqref{dt2DTL:PWF} depends on $k$, $k'$ and $k^{-1}$, $k'^{-1}$ in a covariant way, 
leading to the fact that $u$ and $z$ are dual to each other. 
\end{remark}

The bilinear equation \eqref{dt2DTL:BL} and the modified equation
\eqref{dt2DTL:Mod} were originally given by Date, Jimbo and Miwa \cite{DJM82} 
within the framework of transformation groups for soliton equations.\footnote{
A two-component extension of this modified equation was given in \cite{NCW85}.
}
Through a point transformation $\mathfrak{u}=u-n_1 p_1-n_{-1}/p_{-1}$, the unmodified equation \eqref{dt2DTL:Un} can be written as 
\begin{align*}
 \frac{(\mathfrak{u}-\check{\mathfrak{u}})(\dot{\check{\mathfrak{u}}}-\check{\wt{\mathfrak{u}}})}
 {(\mathfrak{u}-\underdot{\wt{\mathfrak{u}}})(\wt{\mathfrak{u}}-\check{\wt{\mathfrak{u}}})}=1,
\end{align*}
which appeared very recently in \cite{FNR15} (after a recombination of all the discrete shifts). 
The parametrisation in \eqref{dt2DTL:Un}, however, describes its solution structure in a more natural way and allows to consider its continuum limit. 
Furthermore, by a transform $v=\exp\varphi$ we obtain from \eqref{dt2DTL:Mod}
\begin{align}\label{dt2DTL}
 p_1p_{-1}\left(\exp(\check{\wt\varphi}-\check\varphi)-\exp(\wt\varphi-\varphi)\right)
 =-\exp(\dot{\check\varphi}-\varphi)+\exp(\check{\wt\varphi}-\wt{\underdot\varphi}),
\end{align}
which is the discrete analogue of \eqref{2DTL}.

\subsection{Continuum limits}\label{S:Limit}

We set up the continuum limit scheme for the discrete-time 2DTL of $A_\infty$-type. The continuous independent variables are introduced as follows: 
\begin{align*}
 x_1=\frac{n_1}{p_1}, \quad x_{-1}=\frac{n_{-1}}{p_{-1}}, \quad n_1,n_{-1}\rightarrow \infty, \quad p_1,p_{-1}\rightarrow \infty.
\end{align*}
To respect the tradition, we mark the discrete variable $n$ explicitly as a suffix in the corresponding variables. 
The continuum limit scheme results in the maps between the discrete and continuous spaces, namely for $f=f_{n_1,n_{-1},n}=f_n(x_1,x_{-1})$ we have 
\begin{align*}
 &\wt f=f_{n}(x_1+1/p_1,x_{-1}), \quad \check f=f_n(x_1,x_{-1}+1/p_{-1}), \quad \dot f=f_{n+1}(x_1,x_{-1}). 
\end{align*}
as well as
\begin{align*}
 \ut f=f_n(x_1-1/p_1,x_{-1}), \quad \undercheck f=f_n(x_1,x_{-1}-1/p_{-1}), \quad \underdot f=f_{n-1}(x_1,x_{-1}). 
\end{align*}
where $f$ can be any of the variables $u$, $v$ and $\tau$. 

In the continuum limit, the discrete bilinear equation \eqref{dt2DTL:BL} turns out to be
\begin{align*}
 \frac{1}{2}D_{1}D_{-1}\tau_n\cdot\tau_n=\tau_n^2-\tau_{n+1}\tau_{n-1},
\end{align*}
where $D_1$ and $D_{-1}$ are Hirota's bilinear operators\footnote{
For given differentiable functions $f(x)$ and $g(x)$, Hirota's bilinear derivative with respect to $x$ is defined by 
\begin{align*}
 D_x f\cdot g=(\partial_x-\partial_{x'})f(x)g(x')|_{x'=x}.
\end{align*}
}
with respect to $x_1$ and $x_{-1}$. This is the bilinear 2DTL, see e.g. \cite{Hir81,DJM82}. 
Similarly we also have the continuum limits of the unmodified and modified equations as follows: 
\begin{align*}
 \partial_1\ln(1-\partial_{-1}u_n)=u_{n+1}-2u_n+u_{n-1}, \quad
 \partial_1\partial_{-1}\ln v_n=-\frac{v_{n+1}}{v_n}+\frac{v_n}{v_{n-1}}.
\end{align*}
These equations are different nonlinear forms of the 2DTL \eqref{2DTL}, 
as we can still see the characteristic of the Cartan matrix corresponding to $A_\infty$ in both equations. 
By transform $v_n=\exp\varphi_n$, the $v_n$ equation becomes \eqref{2DTL}, 
which can alternatively be derived by taking the limit of \eqref{dt2DTL}. 
All these potentials are connected with each other via transforms 
\begin{align*}
 \varphi_n=\ln v_n=\ln\frac{\tau_{n+1}}{\tau_n} \quad \hbox{and} \quad 
 u_{n+1}-u_n=\partial_1\ln v_n=\partial_1\ln\frac{\tau_{n+1}}{\tau_n},
\end{align*}
which are actually the continuum limits of the bilinear transforms \eqref{dt2DTL:BLT} and Miura transforms \eqref{dt2DTL:MT}. 
The $u_n$ equation can also be written in a slightly different form 
\begin{align*}
 \partial_1\partial_{-1}\ln(1-s_n)=s_{n+1}-2s_n+s_{n-1}
\end{align*}
via transform $s_n=\partial_{-1}u_n$, which is the nonpotential form of the $u_n$ equation and was the form considered in \cite{HIK88}. 

The Lax pairs of the continuous equations can be recovered from the discrete ones in the same limit scheme. 
One can also take the limit only with respect to $n_1$ or $n_{-1}$, from which the semi-discrete equations will arise. 

\section{Discrete-time two-dimensional Toda lattices of $A_{r-1}^{(1)}$-type}\label{S:Reduction}

\subsection{Periodic reductions}\label{S:Periodic}

Performing the $r$-periodic reduction (for integer $r\geq2$) of the discrete-time 2DTL of $A_\infty$-type is equivalent to 
considering sub-algebra $A_{r-1}^{(1)}$, see e.g. \cite{JM83}. In the direct linearisation framework, 
such a reduction can be realised by taking the measure 
\begin{align}\label{Reduction}
 \rd\zeta(k,k')=\sum_{j\in J}\rd\lambda_j(k)\rd k'\delta(k'+\oa^j k),
\end{align}
where $\rd\ld_j(k)$ are the measures only depending on the spectral variable $k$, 
and $\oa=\exp(2\pi\mathfrak{i}/r)$ and $J=\{j|\hbox{$0<j<r$ are integers coprime to $r$}\}$, i.e. $\oa^j$ are all $r$th primitive roots of unity. 
In other words, a constraint is imposed on the two spectral parameters, restricting $k$ and $k'$ on an algebraic curve $k^r=(-k')^r$, cf. \cite{FN17b}. 
As a consequence, the linear integral equation \eqref{Integral} associated with a double integral degenerates, 
and becomes one with only a single integral, namely 
\begin{align}\label{Reduction:Integral}
 \bu_k+\sum_{j\in J}\int_{\Gamma_j}\rd\ld_j(l)\rho_k\Oa_{k,-\oa^j l}\sigma_{-\oa^j l}\bu_l=\rho_k\bc_k, 
\end{align}
where $\Gamma_j$ are the corresponding contours. Meanwhile, reduction \eqref{Reduction} also results in 
\begin{align}\label{Reduction:U}
 \bC=\sum_{j\in J}\int_{\Gamma_j}\rd\ld_j(k)\rho_k\bc_k\tbc_{-\oa^j k}\sigma_{-\oa^j k} \quad \hbox{and} \quad 
 \bU=\sum_{j\in J}\int_{\Gamma_j}\rd\ld_j(k)\bu_k\tbc_{-\oa^j k}\sigma_{-\oa^j k},
\end{align}
where $\rho_k$ and $\sigma_{k'}$ still take their respective forms given in \eqref{dt2DTL:Linear}, 
but one now has to keep in mind that the reduced effective plane wave factors become 
\begin{align}\label{Reduction:PWF}
 \rho_k\sigma_{-\oa^j k}=\left(\frac{p_1+k}{p_1+\oa^j k}\right)^{n_1}\left(\frac{p_{-1}+k^{-1}}{p_{-1}+(\oa^j k)^{-1}}\right)^{n_{-1}}
 \left(\frac{1}{\oa^j}\right)^n.
\end{align}
Thus, the variables $\phi=(\bu_k)_0$, $u=(\bU)_{0,0}$, $v=1-(\bU)_{0,-1}$ and $\tau=\det(1+\bOa\cdot\bC)$ governed by \eqref{Reduction:Integral} 
and \eqref{Reduction:U} play the roles of the direct linearising solutions to the corresponding linear and nonlinear equations. 
For explicit solitons, we can still take particular measures which bring a finite number of poles as we have done in subsection \ref{S:Soliton}. 
We omit the derivation here since the structure of the direct linearising solution is already clear 
and the resulting soliton structure is very similar to that for the 3D case in theorem \ref{T:dt2DTLSol}. 
The only comment here is that in this case a block Cauchy matrix structure will arise, cf. e.g. \cite{ZZN12}, 
since all the primitive roots are involved. As the simplest example, one can take $k'_j=-\oa k_j$ (i.e. we only consider one primitive root) 
in theorem \ref{T:dt2DTLSol}. This immediately provides us with the soliton solutions of the reduced equations.

\subsection{Constraints on the linear and nonlinear variables}\label{S:Constraints}

The reduction \eqref{Reduction} provides not only the reduced direct linearising solutions, 
but also the corresponding constraints on the linear, bilinear and nonlinear variables in the 3D theory, leading to 2D equations. 
In fact, we can observe from \eqref{Reduction:PWF} that the plane wave factors obeys $\rho_k(n+r)\sigma_{-\oa^j k}(n+r)=\rho_k(n)\sigma_{-\oa^j k}(n)$, 
since $\oa^j$ are the $r$th primitive roots of unity, which implies that the reduced infinite matrix $\bC$ given in \eqref{Reduction:U} satisfies 
$\bC(n+r)=\bC(n)$ as its dynamics rely on the effective plane wave factors $\rho_k\sigma_{-\oa^j k}$. 
Recalling that the dynamics of the infinite matrix $\bU$ and the tau function $\tau$ explicitly depend on $\bC$ 
and the dynamics of the wave function $\bu_k$ are determined by $\bU$ and $\rho_k$, cf. \eqref{U}, \eqref{tau} and \eqref{uk}, 
we can therefore obtain 
\begin{align*}
 \bU(n+r)=\bU(n), \quad \tau(n+r)=\tau(n) \quad \hbox{and} \quad \bu_k(n+r)=k^r\bu_k(n).
\end{align*}
For future convenience, in the discrete-time 2DTL of $A_{r-1}^{(1)}$-type, we introduce a suffix $n$ for each variable. 
For example, we mark $u=u_n$, $\dot u=u_{n+1}$ and $\underdot u=u_{n-1}$ (and similar for $v$, $\tau$ and $\phi$). 
Taking the corresponding components or entries for the above relations, we have the constraints for all the variables we need 
and conclude them in the following proposition. 
\begin{proposition}\label{P:Constraints}
The $r$-periodic reduction results in the following constraints on the nonlinear, bilinear and linear variables: 
\begin{align}\label{Constraints}
 u_{n+r}=u_n, \quad v_{n+r}=v_{n}, \quad \tau_{n+r}=\tau_n, \quad \phi_{n+r}=k^n\phi_n,
\end{align}
which implies that we only need to consider $u_i$, $v_i$, $\tau_i$ and $\phi_i$ for $i=0,1,\cdots,r-1$ in the $r$-periodic reduction. 
\end{proposition}
Furthermore, there are also additional relations arising from the reduction. 
\begin{proposition}\label{P:Conservation}
The unmodified and modified variables $u_n$ and $v_n$ obey identities 
\begin{align}\label{Conservation}
 \prod_{i=0}^{r-1}(p_1+u_{i+1}-\wt u_i)=p_1^r, \quad \prod_{i=0}^{r-1}\left[1+p_{-1}(u_i-\check u_i)\right]=1 \quad \hbox{and} \quad \prod_{i=0}^{r-1}v_i=1, 
\end{align}
in the class of the discrete-time 2DTL of $A_{r-1}^{(1)}$-type. 
\end{proposition}
\begin{proof}
These identities are easily proven as they all follow from \eqref{dt2DTL:BLT} and proposition \ref{P:Constraints}. For instance, 
\begin{align*}
 \prod_{i=0}^{r-1}(p_1+u_{i+1}-\wt u_i)=p_1^r\prod_{i=0}^{r-1}\frac{\tau_i\wt\tau_{i+1}}{\tau_{i+1}\wt\tau_i}
 =p_1^r\frac{\tau_0\wt\tau_r}{\tau_r\wt\tau_0}=p_1^r \quad \hbox{and} \quad
 \prod_{i=0}^{r-1}v_n=\prod_{i=0}^{r-1}\frac{\tau_{i+1}}{\tau_i}=\frac{\tau_r}{\tau_0}=1,
\end{align*}
where $\tau_r=\tau_0$ is used in the last step in each equation. 
The second identity is proven similarly. 
\end{proof}
The identities in proposition \ref{P:Conservation} were also given in \cite{FX17} and were referred to as discrete first integrals. 
Here we can see that they are consequences of the periodicity of the tau function in the reduction.

\subsection{Reduced bilinear and nonlinear equations}\label{S:NL}

We can now construct the bilinear and nonlinear equations in the class of the discrete-time 2DTL of $A_{r-1}^{(1)}$-type, 
with the help of the constraints generated by the periodic reduction. 

We first consider the bilinear equation, which is derived from \eqref{dt2DTL:BL}. 
By selecting $n=0,1,\cdots,r-1$ and making use of $\tau_{n+r}=\tau_n$, we obtain a coupled system of bilinear discrete equations 
\bse\label{Reduction:BL}
\begin{align}
 &p_1p_{-1}(\tau_0\check{\wt\tau}_0-\wt\tau_0\check\tau_0)=\tau_0\check{\wt\tau}_0-\check\tau_{1}\wt\tau_{r-1}, \\
 &p_1p_{-1}(\tau_n\check{\wt\tau}_n-\wt\tau_n\check\tau_n)=\tau_n\check{\wt\tau}_n-\check\tau_{n+1}\wt\tau_{n-1}, \quad n=1,2,\cdots,r-2, \\
 &p_1p_{-1}(\tau_{r-1}\check{\wt\tau}_{r-1}-\wt\tau_{r-1}\check\tau_{r-1})=\tau_{r-1}\check{\wt\tau}_{r-1}-\check\tau_{0}\wt\tau_{r-2},
\end{align}
\ese
which we refer to as the bilinear discrete-time 2DTL of $A_{r-1}^{(1)}$-type. 

Similarly, the unmodified discrete-time 2DTL of $A_{r-1}^{(1)}$-type is derived from \eqref{dt2DTL:Un} with the help of $u_{n+r}=u_n$ 
and also takes a coupled system form 
\bse\label{Reduction:Un}
\begin{align}
 &\frac{p_1+\check u_{1}-\check{\wt u}_0}{p_1+u_0-\wt u_{r-1}}=\frac{1+p_{-1}(\wt u_0-\check{\wt u}_0)}{1+p_{-1}(u_0-\check u_0)}, \\
 &\frac{p_1+\check u_{n+1}-\check{\wt u}_n}{p_1+u_n-\wt u_{n-1}}=\frac{1+p_{-1}(\wt u_n-\check{\wt u}_n)}{1+p_{-1}(u_n-\check u_n)}, \quad n=1,2,\cdots,r-2, \\
 &\frac{p_1+\check u_{0}-\check{\wt u}_{r-1}}{p_1+u_{r-1}-\wt u_{r-2}}=\frac{1+p_{-1}(\wt u_{r-1}-\check{\wt u}_{r-1})}{1+p_{-1}(u_{r-1}-\check u_{r-1})}.
\end{align}
\ese
This system should also be referred to as the negative flow\footnote{\label{dBSQ:PWF}
Compare \eqref{Reduction:PWF} with the effective plane wave factor 
\begin{align*}
 \rho_k\sigma_{-\oa^j k}=\left(\frac{p_1+k}{p_1+\oa^j k}\right)^{n_1}\left(\frac{p_2+k}{p_2+\oa^j k}\right)^{n_2}
\end{align*}
for the discrete GD hierarchy, see \cite{NPCQ92}. 
} of the $r$th member in the discrete unmodified GD hierarchy. 
Although the variable $u$ obeys an additional identity (see proposition \ref{P:Conservation}), 
in this case it seems not possible to eliminate one dependent variable in \eqref{Reduction:Un} and express the coupled system by $r-1$ components. 
However, for the positive flow of a member in the discrete GD hierarchy, 
one can reduce the number of components by one, making use of such an identity, see \cite{FX17}. 

The modified discrete-time 2DTL of $A_{r-1}^{(1)}$-type can be written as the following coupled system of discrete equations: 
\bse\label{Reduction:Mod}
\begin{align}
 &p_1p_{-1}\left(\frac{\check{\wt v}_0}{\check v_0}-\frac{\wt v_0}{v_0}\right)
 =-\frac{\check v_{1}}{v_0}+\frac{\check{\wt v}_0}{\wt v_{r-1}}, \\
 &p_1p_{-1}\left(\frac{\check{\wt v}_n}{\check v_n}-\frac{\wt v_n}{v_n}\right)
 =-\frac{\check v_{n+1}}{v_n}+\frac{\check{\wt v}_n}{\wt v_{n-1}}, \quad n=1,2,\cdots,r-2, \\
 &p_1p_{-1}\left(\frac{\check{\wt v}_{r-1}}{\check v_{r-1}}-\frac{\wt v_{r-1}}{v_{r-1}}\right)
 =-\frac{\check v_{0}}{v_{r-1}}+\frac{\check{\wt v}_{r-1}}{\wt v_{r-2}},
\end{align}
\ese
which can also be referred to as the negative flow of the $r$th member in the discrete modified GD hierarchy. 
In this case, it is possible to eliminate one of the components with the help of proposition \ref{P:Conservation}. 
Without loss of generality, we can replace $v_{r-1}$ by $v_0,\cdots,v_{r-2}$ according to $v_{r-1}=1/\prod_{i=0}^{r-2}v_i$. 
Thus, the modified system can also be written as 
\begin{align*}
 &p_1p_{-1}\left(\frac{\check{\wt v}_0}{\check v_0}-\frac{\wt v_0}{v_0}\right)
 =-\frac{\check v_{1}}{v_0}+\check{\wt v}_0\prod_{i=0}^{r-2}\wt v_i, \\
 &p_1p_{-1}\left(\frac{\check{\wt v}_n}{\check v_n}-\frac{\wt v_n}{v_n}\right)
 =-\frac{\check v_{n+1}}{v_n}+\frac{\check{\wt v}_n}{\wt v_{n-1}}, \quad n=1,2,\cdots,r-3, \\
 &p_1p_{-1}\left(\frac{\check{\wt v}_{r-2}}{\check v_{r-2}}-\frac{\wt v_{r-2}}{v_{r-2}}\right)
 =-\frac{1}{v_{r-2}\prod_{i=0}^{r-2}\check v_i}+\frac{\check{\wt v}_{r-2}}{\wt v_{r-3}},
\end{align*}
The transform $v_n=\exp\varphi_n$ gives us the exponential form of the discrete-time 2DTL of $A_{r-1}^{(1)}$-type. 

The constraints \eqref{Constraints} simultaneously bring us the Miura and bilinear transforms 
\bse\label{Reduction:Tranforms}
\begin{align}
 &p_1+u_{n+1}-\wt u_n=p_1\frac{\wt v_n}{v_n}=p_1\frac{\tau_n\wt\tau_{n+1}}{\tau_{n+1}\wt\tau_n}, \quad 
 p_1+u_0-\wt u_{r-1}=p_1\frac{\wt v_{r-1}}{v_{r-1}}=p_1\frac{\tau_{r-1}\wt\tau_{0}}{\tau_{0}\wt\tau_{r-1}}
\end{align}
for $n=0,1,\cdots,r-2$, and 
\begin{align}
 &1+p_{-1}(u_0-\check u_0)=\frac{\check v_0}{v_{r-1}}=\frac{\tau_{r-1}\check\tau_{1}}{\tau_0\check\tau_0}, \nonumber \\
 &1+p_{-1}(u_n-\check u_n)=\frac{\check v_n}{v_{n-1}}=\frac{\tau_{n-1}\check\tau_{n+1}}{\tau_n\check\tau_n}, \\
 &1+p_{-1}(u_{r-1}-\check u_{r-1})=\frac{\check v_{r-1}}{v_{r-2}}=\frac{\tau_{r-2}\check\tau_{0}}{\tau_{r-1}\check\tau_{r-1}} \nonumber
\end{align}
\ese
for $n=1,2,\cdots,r-2$, which relate \eqref{Reduction:Un}, \eqref{Reduction:Mod} and \eqref{Reduction:BL} with each other. 

The unmodified equation \eqref{Reduction:Un} and the modified equation \eqref{Reduction:Mod} 
are new parametrisations of the integrable lattice equations corresponding to equivalence class $[(0,1;r-1,0)]$ given in \cite{FX17}, 
allowing exact solutions and continuum limits.

\subsection{Lax matrices}\label{S:Lax}

The Lax pairs for the reduced equations \eqref{Reduction:BL}, \eqref{Reduction:Un} and \eqref{Reduction:Mod} 
can be constructed from \eqref{dt2DTL:Lax} under the constraints given in proposition \ref{P:Constraints}. 
We only write down the Lax pair for the unmodified equation in $u_n$, 
and those for the modified and bilinear equations can be obtained by using the Miura and bilinear transforms given in \eqref{Reduction:Tranforms}. 

The ``tilde'' equation of the Lax pair is a consequence of \eqref{dt2DTL:Laxa} and the constraints on $u_n$ and $\phi_n$ in \eqref{Constraints}, 
taking the form of 
\bse\label{Reduction:Lax}
\begin{align}\label{Reduction:Laxa}
 \left(
 \begin{array}{c}
 \phi_0 \\
 \phi_1 \\
 \vdots \\
 \phi_{r-2} \\
 \phi_{r-1} 
 \end{array}
 \right)^{\tilde{}}=
 \begin{pmatrix}
  p_1+u_1-\wt u_0 & 1 & & & \\
  & p_1+u_2-\wt u_1 & 1 & & \\
  & & \ddots & \ddots & \\
  & & & p_1+u_{r-1}-\wt u_{r-2} & 1 \\
  k^r & & & & p_1+u_0-\wt u_{r-1}
 \end{pmatrix}
 \left(
 \begin{array}{c}
 \phi_0 \\
 \phi_1 \\
 \vdots \\
 \phi_{r-2} \\
 \phi_{r-1} 
 \end{array}
 \right).
\end{align}
This linear equation is also compatible with the positive flows of members in the discrete GD hierarchy, 
and it is gauge equivalent to the one given in \cite{NPCQ92}. 
The dynamics describing the negative flow is contained in the ``check'' part, 
which follows from \eqref{dt2DTL:Laxb} and $\phi_{-1}=k^{-r}\phi_{r-1}$ in proposition \ref{P:Constraints} and is given by 
\begin{align}\label{Reduction:Laxb}
 \left(
 \begin{array}{c}
 \phi_0 \\
 \phi_1 \\
 \vdots \\
 \phi_{r-2} \\
 \phi_{r-1} 
 \end{array}
 \right)^{\check{}}=
 \begin{pmatrix}
  p_{-1} & & & & * \\
  1+p_{-1}(u_1-\check u_1) & p_{-1} & & & \\
  & \ddots & \ddots & & \\
  & & 1+p_{-1}(u_{r-2}-\check u_{r-2}) & p_{-1} & \\
  & & & 1+p_{-1}(u_{r-1}-\check u_{r-1}) & p_{-1}
 \end{pmatrix}
 \left(
 \begin{array}{c}
 \phi_0 \\
 \phi_1 \\
 \vdots \\
 \phi_{r-2} \\
 \phi_{r-1} 
 \end{array}
 \right),
\end{align}
\ese
where $*=k^{-r}[1+p_{-1}(u_0-\check u_0)]$. 
The compatibility condition of the Lax pair gives rise to the coupled system \eqref{Reduction:Un}. 

An alternative way to construct such a linear problem was discussed within a framework of $\mathbb{Z}_N$ graded algebras. 
In fact, by introducing a gauge transform 
\begin{align*}
 \left(
 \begin{array}{c}
 \psi_0 \\
 \psi_1 \\
 \vdots \\
 \psi_{r-2} \\
 \psi_{r-1} 
 \end{array}
 \right)
 =
 \begin{pmatrix}
  1 & & & & \\
  & k^{-1} & & & \\
  & & \ddots & & \\
  & & & k^{-(r-2)} & \\
  & & & & k^{-(r-1)}
 \end{pmatrix}
 \left(
 \begin{array}{c}
 \phi_0 \\
 \phi_1 \\
 \vdots \\
 \phi_{r-2} \\
 \phi_{r-1} 
 \end{array}
 \right),
\end{align*}
the linear equations in \eqref{Reduction:Lax} become more or less the same as the ones given in \cite{FX17}. 
To be more precise, the ``tilde'' part becomes exactly the same (up to a translation on the potential variables $u_n$), 
and the ``check'' part has the same matrix structure but in our case it depends on $k^{-1}$ explicitly.

\subsection{$A_1^{(1)}$ and $A_2^{(1)}$: Negative flows of discrete Korteweg--de Vries and Boussinesq}\label{S:Examples}

We list two concrete examples explicitly, namely the discrete-time 2DTLs of $A_1^{(1)}$-type and $A_2^{(1)}$-type. 
The $A_1^{(1)}$ class is corresponding to the negative flows of the discrete KdV-type equations, 
including the discrete sine--Gordon equation which was discovered and rediscovered in the literature for times. 
The $A_2^{(2)}$ class gives the negative flows of the discrete Boussinesq-type equations. 

The $A_1^{(1)}$ class includes the negative flows of the discrete unmodified and modified KdV equations, 
and the bilinear discrete-time 2DTL of $A_1^{(1)}$ as follows: 
\begin{align*}
 &\frac{p_1+\check u_1-\check{\wt u}_0}{p_1+u_0-\wt u_1}=\frac{1+p_{-1}(\wt u_0-\check{\wt u}_0)}{1+p_{-1}(u_0-\check u_0)}, \quad 
 \frac{p_1+\check u_0-\check{\wt u}_1}{p_1+u_1-\wt u_0}=\frac{1+p_{-1}(\wt u_1-\check{\wt u}_1)}{1+p_{-1}(u_1-\check u_1)}, \\
 &p_1p_{-1}\left(\frac{\check{\wt v}_0}{\check v_0}-\frac{\wt v_0}{v_0}\right)=-\frac{1}{v_0\check v_0}+\wt v_0\check{\wt v}_0, \\
 &p_1p_{-1}(\tau_0\check{\wt\tau}_0-\wt\tau_0\check\tau_0)=\tau_0\check{\wt\tau}_0-\check\tau_1\wt\tau_1, \quad 
 p_1p_{-1}(\tau_1\check{\wt\tau}_1-\wt\tau_1\check\tau_1)=\tau_1\check{\wt\tau}_1-\check\tau_0\wt\tau_0.
\end{align*}
The transform $v_0=\exp\varphi_0$ then brings us 
\begin{align*}
 p_1p_{-1}\left(\exp(\check{\wt\varphi}_0-\check\varphi_0)-\exp(\wt\varphi_0-\varphi_0)\right)
 =\exp(\wt\varphi_0+\check{\wt\varphi}_0)-\exp(-\varphi_0-\check\varphi_0),
\end{align*}
which is the discrete analogue of the continuous-time sinh--Gordon equation (i.e. the first equation in \eqref{sinh-Gordon}). 
This equation can also be written in another form, namely the discrete sine--Gordon equation
\begin{align*}
 p_1p_{-1}\sin(\vartheta_0+\check{\wt\vartheta}_0-\wt\vartheta_0-\check\vartheta_0)=\sin(\vartheta_0+\wt\vartheta_0+\check\vartheta_0+\check{\wt\vartheta}_0),
\end{align*}
by transform $\varphi_0=2\mathfrak{i}\vartheta_0$. One comment here is that the unmodified equation in $u_0$ and $u_1$ can also be written in a coupled system 
composed of the unmodified and modified variables $u_0$ and $v_0$, with the help of the Miura transform $p_1+u_1-\wt u_0=p_1\wt v_0/v_0$. 

Now we consider the class of the discrete-time 2DTL of $A_2^{(1)}$-type, namely the 3-periodic reduction. 
The bilinear discrete-time 2DTL of $A_2^{(1)}$-type according to the general framework \eqref{Reduction:BL} is given by 
\begin{align*}
 p_1p_{-1}(\tau_0\check{\wt\tau}_0-\wt\tau_0\check\tau_0)=\tau_0\check{\wt\tau}_0-\check\tau_1\wt\tau_2, \quad
 p_1p_{-1}(\tau_1\check{\wt\tau}_1-\wt\tau_1\check\tau_1)=\tau_1\check{\wt\tau}_1-\check\tau_2\wt\tau_0, \quad
 p_1p_{-1}(\tau_2\check{\wt\tau}_2-\wt\tau_2\check\tau_2)=\tau_2\check{\wt\tau}_2-\check\tau_0\wt\tau_1,
\end{align*}
which is also the negative flow of the trilinear discrete Boussinesq equation discussed in \cite{FN17b,ZZN12}. 

The unmodified equation in this class is a three-component system involving $u_0$, $u_1$ and $u_2$, which takes the form of 
\begin{align*}
 \frac{p_1+\check u_1-\check{\wt u}_0}{p_1+u_0-\wt u_2}=\frac{1+p_{-1}(\wt u_0-\check{\wt u}_0)}{1+p_{-1}(u_0-\check u_0)}, \quad 
 \frac{p_1+\check u_2-\check{\wt u}_1}{p_1+u_1-\wt u_0}=\frac{1+p_{-1}(\wt u_1-\check{\wt u}_1)}{1+p_{-1}(u_1-\check u_1)}, \quad 
 \frac{p_1+\check u_0-\check{\wt u}_2}{p_1+u_2-\wt u_1}=\frac{1+p_{-1}(\wt u_2-\check{\wt u}_2)}{1+p_{-1}(u_2-\check u_2)}.
\end{align*}
This equation is also not decoupled and acts as the negative flow of the 9-point discrete unmodified Boussinesq equation proposed in \cite{NPCQ92}. 

The modified equation is also a three-component system of discrete equations, composed of $v_0$, $v_1$ and $v_2$, according to \eqref{Reduction:Mod}. 
In this case, we can make use of the identity $v_2=1/(v_0v_1)$ (cf. equation \eqref{Conservation}) and eliminate the component $v_2$. 
As a consequence, we have 
\begin{align*}
 p_1p_{-1}\left(\frac{\check{\wt v}_0}{\check v_0}-\frac{\wt v_0}{v_0}\right)
 =-\frac{\check v_1}{v_0}+\check{\wt v}_0\wt v_0\wt v_1, \quad 
 p_1p_{-1}\left(\frac{\check{\wt v}_1}{\check v_1}-\frac{\wt v_1}{v_1}\right)
 =-\frac{1}{\check v_0\check v_1 v_1}+\frac{\check{\wt v}_1}{\wt v_{0}},
\end{align*}
which is the negative flow of the 9-point discrete modified Boussinesq equation given in \cite{NPCQ92}. 
The exponential form of the discrete-time 2DTL of $A_2^{(1)}$-type is derived by the transforms $v_0=\exp\varphi_0$ and $v_1=\exp\varphi_1$, 
taking the form of 
\begin{align*}
 &p_1p_{-1}\left(\exp(\check{\wt\varphi}_0-\check\varphi_0)-\exp(\wt\varphi_0-\varphi_0)\right)
 =\exp(\check{\wt\varphi}_0+\wt\varphi_0+\wt\varphi_1)-\exp(\check\varphi_1-\varphi_0), \\
 &p_1p_{-1}\left(\exp(\check{\wt\varphi}_1-\check\varphi_1)-\exp(\wt\varphi_1-\varphi_1)\right)
 =\exp(\check{\wt\varphi}_1-\wt\varphi_0)-\exp(-\check\varphi_0-\check\varphi_1-\varphi_1),
\end{align*}
which is the $A_2^{(1)}$-type analogue of \eqref{dt2DTL}. 

The continuum limits of these discrete equations follow the same scheme in subsection \ref{S:Limit}, 
leading to the negative flows in the continuous KdV and Boussinesq hierarchies.

\section{Concluding remarks}\label{S:Concl}

The direct linearisation scheme was established for the discrete-time 2DTL equations of $A_\infty$-type and $A_{r-1}^{(1)}$-type. 
For each algebra, a class of nonlinear (including bilinear) equations arise and their integrability is guaranteed 
in the sense of having Lax pairs and direct linearising solutions. 

For convenience we mainly focused on equations expressed by the (potential) unmodified variable, the (potential) modified variable, 
as well as the tau function, because the resulting equations take relatively simple forms (i.e. scalar octahedron-type equations). 
There certainly exist alternative nonlinear forms, and sometimes a nonlinear equation could even take the form of a coupled system, 
which normally happens in a closed-form equation expressed by a nonpotential variable, see e.g. \cite{HIK88,KOS94}. 

In fact, in the framework there is also another nonpotential form of the discrete-time 2DTL of $A_\infty$-type. 
Introducing nonpotential variables
\begin{align*}
 P\doteq p_{-1}^{-1}+u-\check u=p_{-1}^{-1}\frac{\check v}{\underdot v}=p_{-1}^{-1}\frac{\underdot\tau\dot{\check\tau}}{\tau\check\tau} \quad \hbox{and} \quad 
 Q\doteq p_1+\dot u-\wt u=p_1\frac{\wt v}{v}=p_1\frac{\tau\dot{\wt\tau}}{\dot\tau\wt\tau},
\end{align*}
we obtain the nonpotential equation taking the form of a coupled system of $P$ and $Q$, namely 
\begin{align}\label{dt2DTL:Nonpotential}
 \check Q/\underdot Q=\wt P/P, \quad \check Q-Q=\wt P-\dot P,
\end{align}
which follow from the unmodified equation \eqref{dt2DTL:Un} and the modified equation \eqref{dt2DTL:Mod}, respectively. 
Equation \eqref{dt2DTL:Nonpotential} is still a 6-point equation (if we count both components $P$ and $Q$), 
similar to equations given in theorem \ref{T:dt2DTLNL}. The Lax pair of this coupled system takes the form of 
\begin{align*}
 \wt\phi=Q\phi+\dot\phi, \quad \check\phi=p_{-1}(\phi+P\underdot\phi)
\end{align*}
which is derived from \eqref{dt2DTL:Lax} by replacing $u$ by the nonpotential variables. 
Following the idea about deriving the nonpotential discrete KP equation given in \cite{FN17a},
one can eliminate either $P$ or $Q$ in \eqref{dt2DTL:Nonpotential}.
As a consequence, a 10-point scalar equation in terms of only $Q$ or $P$ can be expected, which we leave as an exercise. 
This is not surprising as we have already seen such a lattice structure in the derivation of the bilinear equation \eqref{dt2DTL:BL}, 
see the proof of theorem \ref{T:dt2DTLNL}. 

The solution structure of the discrete-time 2DTL of $A_\infty$-type is the related to that of the discrete KP equation, cf. \cite{FN17a}. 
To be more precise, the nonlinear structure \eqref{dt2DTL:Kernel} is exactly the same, and the linear structure \eqref{dt2DTL:PWF} 
is a slight deformation of that in KP and it can be reconstructed from the plane wave factor of the discrete KP equation in a subtle way. 
However, such a deformation leads to discrete integrable systems having different lattice structures 
(though in this case they are still octahedron-type equations). 

The unmodified and modified equations \eqref{Reduction:Un} and \eqref{Reduction:Mod} 
are corresponding to some new equations equivalence class $[(0,1;r-1,0)]$ in \cite{FX17}. 
Thus, we identify that these new equations are the nonlinear equations in the class of the discrete-time 2DTL of $A_{r-1}^{(1)}$-type, 
which arise as the dimensional reductions of \eqref{dt2DTL:Un} and \eqref{dt2DTL:Mod}. 
Equations in other equivalence classes in \cite{FX17} are actually also associated with the algebra $A_{r-1}^{(1)}$. To put it another way, 
these new equations share the same kernel and measure in the direct linearisation, and obey different (discrete) time evolutions. 

The direct linearisation of the discrete-time 2DTLs of other types (such as $B_\infty$ and $C_\infty$ and their reductions) will be reported elsewhere. 
From the author's experience, the structures of these equations are very different. 
This is mainly because of the fractionally linear dependence on the spectral parameters in their dispersions \cite{FN17a}. 
However, such an issue does not occur in the continuous theory, and various continuous-time 2DTLs of different types arise 
very naturally as reductions of the 2DTL of $A_\infty$-type. 

\subsection*{Acknowledgements}
The author is very grateful to both referees for their suggestions, which improve the manuscript a lot. 
He also thanks Frank Nijhoff for useful comments on a draft version of the manuscript 
and Allan Fordy for explaining some details in \cite{FX17}. 
This project was supported by a Leeds International Research Scholarship (LIRS) and by the School of Mathematics.

\end{document}